\documentclass[a4paper,UKenglish,cleveref, autoref, thm-restate,authorcolumns]{lipics-v2019}

\def\MdN{\ensuremath{\mathbb{N}}}
\newcommand{\Oh}[1]{\mathcal{O}\!\left( #1\right)}
\newcommand{\Is}{:=}
\newcommand{\etal}{et~al.~}
\newcommand{\ie}{i.e.\ }

\newcommand{\CC}{C\texttt{++}}

\newif\ifEnableExtend
\EnableExtendfalse

\newif\ifDoubleBlind
\DoubleBlindfalse

\usepackage{algorithm}
\usepackage[noend]{algpseudocode}
\usepackage{graphicx}
\usepackage{numprint}
\usepackage[section]{placeins}
\usepackage{subcaption}
\usepackage{wrapfig}
\usepackage{booktabs}

\newcommand*{\shorterDots}{.\kern-0.06em.\kern-0.06em.} 

\newcommand{\cut}{\mathcal{C}}
\newcommand{\bestcut}{\widehat{\cut}}
\newcommand{\wgt}{\mathcal{W}}
\newcommand{\bestwgt}{\widehat{\wgt}}
\newcommand{\vopt}{\mathcal{V}}
\newcommand{\queue}{\mathcal{Q}}
\newcommand{\vcbase}{\texttt{VieCut-MTC}}
\newcommand{\exact}{\texttt{Exact-MTC}}
\newcommand{\inexact}{\texttt{Inexact-MTC}}

\newcommand{\csch}[1]{{\color{red}{[ CS says: #1 ]}}}

\let\oldReturn\Return
\renewcommand{\Return}{\State\oldReturn}

\title{Faster Parallel Multiterminal Cuts}

\ifDoubleBlind
\author{Double Blind}{Double Blind University}{}{}{}
\authorrunning{D. Blind}
\Copyright{Double Blind}
\else\author{Monika Henzinger}{University of Vienna, Faculty of Computer Science, Vienna, Austria}{monika.henzinger@univie.ac.at}{https://orcid.org/0000-0002-5008-6530}{}
\author{Alexander Noe}{University of Vienna, Faculty of Computer Science, Vienna, Austria}{alexander.noe@univie.ac.at}{https://orcid.org/0000-0002-4711-3323}{}
\author{Christian Schulz}{University of Vienna, Faculty of Computer Science, Vienna, Austria}{christian.schulz@univie.ac.at}{https://orcid.org/0000-0002-2823-3506}{}
\authorrunning{M. Henzinger, A. Noe and C. Schulz}
\Copyright{Monika Henzinger, Alexander Noe and Christian Schulz}

\fi{}

\hideLIPIcs  

\EventEditors{John Q. Open and Joan R. Access}
\EventNoEds{2}
\EventLongTitle{42nd Conference on Very Important Topics (CVIT 2016)}
\EventShortTitle{CVIT 2016}
\EventAcronym{CVIT}
\EventYear{2016}
\EventDate{December 24--27, 2016}
\EventLocation{Little Whinging, United Kingdom}
\EventLogo{}
\SeriesVolume{42}
\ArticleNo{23}
\nolinenumbers

\category{}

\relatedversion{}

\ifDoubleBlind

\else
\funding{
  The research leading to these results has received funding from the European Research Council under the
European Community's Seventh Framework Programme (FP7/2007-2013) /ERC grant agreement No. 340506. Partially supported by DFG grant SCHU 2567/1-2. 
We further thank the Vienna Scientific Cluster (VSC) for providing high performance computing resources.
}
\fi{}

\begin{document}

\maketitle

\begin{abstract}
We give an improved branch-and-bound solver for the multiterminal cut problem, based on the recent work of Henzinger~\etal\cite{henzinger2020shared}. We contribute new, highly effective data reduction rules to transform the graph into a smaller equivalent instance. In addition, we present a local search algorithm that can significantly improve a given solution to the multiterminal cut problem.
Our exact algorithm is able to give exact solutions to more and harder problems compared to the state-of-the-art algorithm by Henzinger~\etal\cite{henzinger2020shared}; and give better solutions for more than two third of the problems that are too large to be solved to optimality. Additionally, we give an inexact heuristic algorithm that computes high-quality solutions for very hard instances in reasonable time.
\end{abstract}

\ifDoubleBlind
\setcounter{page}{0}
\fi

\section{Introduction}
The multiterminal cut problem is a fundamental combinatorial optimization problem which was first formulated by Dahlhaus~\etal\cite{dahlhaus1994complexity} and Cunningham~\cite{cunningham1989optimal}. 
Given an undirected edge-weighted graph $G=(V,E,w)$ with edge weights $w: E \mapsto \MdN_{>0}$ and a set $T$, $|T|=k$, of terminals, the \emph{multiterminal cut} problem is to divide its set of nodes into~$k$ blocks such that each block contains exactly one terminal and the weight sum of the edges running between the~blocks~is~minimized. 
There are many applications of the problem: for example multiprocessor scheduling~\cite{DBLP:journals/tse/Stone77}, clustering~\cite{DBLP:journals/njc/PferschyRW94} and bioinformatics~\cite{karaoz2004whole,nabieva2005whole,vazquez2003global}. 

The problem is known to be NP-hard for $k \geq 3$~\cite{dahlhaus1994complexity}.
For $k=2$ the problem reduces to the well known minimum $s$-$t$-cut problem, which is in P. 
The minimum $s$-$t$-cut problem aims to find the minimum cut in which the vertices $s$ and $t$ are in different blocks. Most algorithms for the minimum multiterminal cut problem use minimum s-t-cuts as a subroutine. Dahlhaus~\etal\cite{dahlhaus1994complexity} give a $2(1-1/k)$ approximation algorithm with polynomial running time based on the notion of \emph{minimum isolating cuts}, \ie the minimum cut separating a terminal from all other terminals. 
The currently best known approximation algorithm due to Buchbinder~\etal\cite{buchbinder2013simplex} uses a linear program relaxation to achieve an approximation ratio of $1.323$.
Recently, Henzinger~\etal\cite{henzinger2020shared} introduced a branch-and-reduce framework for the problem that is multiple orders of magnitudes that classic ILP formulations for the problem which has been the de facto standard used by practitioners. This allows researchers to solve instances to optimality that are significantly larger than was previously possible and hence enables the use of multiterminal cut algorithms in practical applications.

\textbf{Contribution.}
We give an improved solver for the multiterminal cut problem, based on the recent work of Henzinger~\etal\cite{henzinger2020shared}. We contribute new, highly effective reductions to transform the graph into a smaller equivalent instance. In addition, we present a local search algorithm that can significantly improve a given solution to the multiterminal cut problem. Additionally, we combine the branch-and-reduce solver with an integer linear program solver to more efficiently solve subproblems emerging over the course of the algorithm. With our newly introduced reducitons, the state-of-the-art algorithm by Henzinger~\etal\cite{henzinger2020shared} is able to solve significantly harder instances to optimality and give better solutions to instances that are too large to solve to optimality. Additionally, we give an inexact algorithm that gives high-quality solutions to hard problems in reasonable time. 

\section{Preliminaries}\label{s:preliminaries}

\subsection{Basic Concepts}
Let $G = (V, E, w)$ be a weighted undirected graph with vertex set $V$, edge set $E \subset V \times V$ and
non-negative edge weights $w: E \rightarrow \mathbb{N}$. 
We extend $w$ to a set of edges $E' \subseteq E$ by summing the weights of the edges; that is, $w(E')\Is \sum_{e=(u,v)\in E'}w(u,v)$ and sets of nodes where $w(V_1, V_2)$ is the sum of edge weights connecting sets $V_1$ and $V_2$. 
Let $n = |V|$ be the
number of vertices and $m = |E|$ be the number of edges in $G$. The \emph{neighborhood}
$N(v)$ of a vertex $v$ is the set of vertices adjacent to $v$. The \emph{weighted degree} of a vertex is the sum of the weights of its incident edges.
For a set of vertices $A\subseteq V$, we denote by $E[A]\Is \{(u,v)\in E \mid u\in A, v\in V\setminus A\}$; that is, the set of edges in $E$ that start in $A$ and end in its complement.
A \emph{$k$-cut}, or \emph{multicut}, is a partitioning of $V$ into $k$ disjoint non-empty blocks, \ie $V_1 \cup \dots \cup V_k = V$. The weight of a $k$-cut is defined as the weight sum of all edges crossing block boundaries, \ie $w(E \cap \bigcup_{i<j}V_i\times V_j)$.

\subsection{Multiterminal Cuts}

A \emph{multiterminal cut} for $k$ terminals $T = \{t_1,\shorterDots,t_k\}$ is a multicut with~$t_1 \in V_1,\shorterDots,t_k \in V_k$. Thus, a multiterminal cut pairwisely separates all terminals from each other. The edge set of the multiterminal cut with minimum weight of $G$ is called $\cut(G)$ and the associated optimal partitioning of vertices is denoted as $\vopt = \{\vopt_1, \dots, \vopt_k\}$. For a vertex $v \in V$, $\vopt(v)$ denotes the block affiliation of $v$ in the optimal partitioning $\vopt$.
 $\cut$ can be seen as the set of all edges that cross block boundaries in $\vopt$, \ie $\cut(G) = \bigcup \{e = (u,v) \mid \vopt_u \neq \vopt_v\}$. The weight of the minimum multiterminal cut is denoted as $\wgt(G) = w(\cut(G))$. At any point in time, the best currently known upper bound for $\wgt(G)$ is denoted as $\bestwgt(G)$ and the best currently known multiterminal cut is denoted as $\bestcut(G)$. If graph $G$ is clear from the context, we omit it in the notation. 
 There may be multiple minimum multiterminal cuts, however, we aim to find one multiterminal cut with minimum weight.

In this paper we use \emph{minimum s-T-cuts}. For a vertex $s$ (\emph{source}) and a non-empty vertex set $T$ (\emph{sinks}), the minimum s-T-cut is the smallest cut in which $s$ is one side of the cut and all vertices in $T$ are on the other side. This is a generalization of minimum s-t-cuts that allows multiple vertices in $T$ and can be easily replaced by a minimum s-t-cut by connecting every vertex in $T$ with a new super-sink by infinite-capacity edges. We denote the capacity of a minimum-s-T-cut, \ie the sum of weights in the smallest cut separating $s$ from $T$, by $\lambda(G,s,T)$. This cut is also called the \emph{minimum isolating cut}~\cite{dahlhaus1994complexity} for vertex $s$ and vertex set $T$ and the minimum isolating cut where the source side is the largest is called the \emph{largest minimum isolating cut} for $s$ and $T$.

In our algorithm we use \emph{graph contraction} and \emph{edge deletions}.
Given
an edge $e = (u, v) \in E$, we define $G/e$ to be the graph  
after \emph{contracting} $e$.
In the contracted graph, we delete vertex $v$ and all incident edges. For each edge $(v, x) \in E$, we add an edge $(u, x)$
with $w(u, x) = w(v, x)$ to $G$ or, if the edge already exists, we give it the edge
weight $w(u,x) + w(v,x)$. For the \emph{edge deletion} of an edge $e$, we define $G-e$ as the graph $G$ in which $e$ has been removed. Other vertices and edges remain the same. An \emph{articulation point} is a vertex whose removal disconnects the graph $G$ into multiple disconnected components. For a given multiterminal cut $S$, the graph $G\backslash S$ splits $G$ into $k$ connected components, called \emph{blocks}, as defined by the cut edges in $S$, each containing exactly one terminal. 

In the last two decades significant advances in FPT algorithms have been made: an NP-hard graph problem is fixed-parameter tractable (FPT) if large inputs can be solved efficiently and provably optimally, as long as some problem parameter is small. This has resulted in an algorithmic toolbox that are by now well-established. 
While the multiterminal cut problem is NP-hard, it is \emph{fixed-parameter tractable} (FPT), parameterized by the multiterminal cut weight $\wgt(G)$. A problem is fixed-parameter tractable if there is a parameter $\sigma$ so that there is an algorithm with runtime $f(\sigma) \cdot n^{\Oh{1}}$. Marx~\cite{marx2006parameterized} proved that the multiterminal cut problem is FPT and Chen~\etal\cite{chen2009improved} gave the first FPT algorithm with a running time of $4^{\wgt(G)}\cdot n^{\Oh{1}}$, later improved by Xiao~\cite{xiao2010simple} to $2^{\wgt(G)}\cdot n^{\Oh{1}}$ and by Cao~\etal\cite{cao20141} to $1.84^{\wgt(G)}\cdot n^{\Oh{1}}$. 
Generally, few of the new FPT techniques are implemented and tested on real datasets, and their practical potential is far from understood. 
However, recently the engineering part in area has gained some momentum. 
There are several experimental studies in the area that take up ideas from FPT or kernelization theory, e.g.~for independent sets (or equivalently vertex cover)~\cite{akiba-tcs-2016,DBLP:conf/sigmod/ChangLZ17,dahlum2016accelerating,DBLP:conf/alenex/Lamm0SWZ19,DBLP:journals/corr/abs-1908-06795,DBLP:conf/alenex/Hespe0S18}, for cut tree construction\cite{DBLP:conf/icdm/AkibaISMY16}, for treewidth computations \cite{bannach_et_al:LIPIcs:2018:9469,DBLP:conf/esa/Tamaki17,koster2001treewidth}, for the feedback vertex set problem \cite{DBLP:conf/wea/KiljanP18,DBLP:conf/esa/FleischerWY09}, for the dominating set problem~\cite{10.1007/978-3-319-55911-7_5}, for the minimum cut~\cite{henzinger2019shared,henzinger2018practical}, for the node ordering problem~\cite{ost2020engineering}, for the maximum cut problem~\cite{DBLP:journals/corr/abs-1905-10902} and for the cluster editing problem~\cite{Boecker2011}.
Recently, this type of data reduction techniques is also applied for problem in P such as matching \cite{DBLP:conf/esa/KorenweinNNZ18},

\subsection{\vcbase}

We present an improved solver for the multiterminal cut problem. Our work is based on a recent result by Henzinger~\etal\cite{henzinger2020shared}, in the following named \vcbase{}. In this section we give a short summary of their results, for further details we refer the reader to their original work~\cite{henzinger2020shared}. The \vcbase{} multiterminal cut solver is a shared-memory parallel solver for the multiterminal cut problem. \vcbase{} is a branch-and-reduce algorithm that performs a set of local contraction routines to transform the graph $G$ into an instance of smaller size $H$, where the minimum multiterminal cut $\wgt(G) = \wgt(H)$, \ie the minimum multiterminal cut $\cut(G)$ can still be found on the smaller instance $H$. For this purpose, they use the following lemmas:

\begin{lemma}\label{lem:cont} \cite{cao20141}\cite{henzinger2020shared}
  If an edge $e = (u,v) \in G$ is guaranteed not to be in at least one multiterminal cut $\cut(G)$ (\ie $\vopt(u) = \vopt(v)$), we can contract $e$ and $\wgt(G/e) = \wgt(G)$. (Proof in~\cite{henzinger2020shared})
\end{lemma}

\begin{lemma}\label{lem:del} \cite{cao20141}\cite{henzinger2020shared}
  If an edge $e = (u,v) \in E$ is guaranteed to be in a minimum multiterminal cut, \ie there is a minimum multiterminal cut $\cut(G)$ in which $\vopt(u) \neq \vopt(v)$, we can delete $e$ from $G$ and $\cut(G - e)$ is still a valid minimum multiterminal cut. (Proof in~\cite{henzinger2020shared})
\end{lemma}

Lemma~\ref{lem:cont} allows the contraction of edges that are guaranteed not to be in at least on multiterminal cut and Lemma~\ref{lem:del} allows the deletion of edges that are guaranteed to be in a multiterminal cut. An example for such an edge is an edge that connects two terminal vertices.

\paragraph*{Largest Minimum Isolating Cut}

Dahlhaus~\etal\cite{dahlhaus1994complexity} show that there exists a minimum multiterminal cut $\cut(G)$ for a graph $G$ such that for every terminal $t \in T$ all vertices on  the source side of the largest minimum isolating cut are in block $t$. Thus, according to Lemma~\ref{lem:cont}, the source sides can be contracted into their respective terminals. The cut value of this problem is equal to the sum of all isolating cuts minus the heaviest, as any set of $t-1$ isolating cuts pairwisely separates all terminals form each other. A lower bound for the optimal solution is the sum of all isolating cuts divided by two~\cite{dahlhaus1994complexity,henzinger2020shared}.

\paragraph*{Reductions}
A variety of reductions in the work of Henzinger~\etal\cite{henzinger2020shared} use Lemma~\ref{lem:cont} to contract edges and effectively reduce the size of the input graph. For every \emph{low degree vertex} $v$ with $N(v) \geq 2$, one can contract the heaviest edge incident to $v$ as there is at least one multiterminal cut that does not contain it. Every \emph{heavy edge} $e = (v,u)$ with $w(e) \cdot 2 \geq w(E[v])$ can also be contracted. This condition can be relaxed to \emph{heavy triangles}, where an edge $e = (v_1, v_2)$ that is part of a triangle $(v_1,v_2,v_3)$ can be contracted if $w(e) + w(v_1, v_3) \cdot 2 \geq w(E[v_1])$ and $w(e) + w(v_2, v_3) \cdot 2 \geq w(E[v_2])$. A more global reduction uses the CAPFOREST algorithm of Nagamochi~\etal\cite{nagamochi1992computing,nagamochi1994implementing} to find a connectivity lower bound of every edge in $G$ in almost linear time. If an edge $e=(u,v)$ has \emph{high connectivity}, \ie there is no small cut that separates $u$ and $v$, and no multiterminal cut that separates its incident vertices can be better than $\bestwgt(G)$, the edge can be contracted according to Lemma~\ref{lem:cont}. For full descriptions and proofs of these reductions we refer the reader to Section $4$ of~\cite{henzinger2020shared}.

\paragraph*{Branching}

When it is not possible to find any more edges to contract or delete, \vcbase{} selects an edge $e$ incident to a terminal and creates two subproblems: $G/e$ represents the problem in which $e$ is not part of the multiterminal cut $\cut(G)$ and $G-e$ represents the problem in which it is. Both subproblems are added to a shared-memory parallel problem queue $\queue$ and solved independently from each other.

\section{Improved Reductions and Branching}

We now introduce a set of new reductions to further decrease the problem size. Additionally, we give an alternative branching rule that allows for faster branching.

\subsection{New Reductions}

\vcbase{} contracts edges incident to low degree vertices, edges with high weight and edges whose incident vertices have a high connectivity. Additionally, \vcbase{} contracts the largest minimum isolating cut for each terminal to the remainder of the terminal set. We now introduce additional reductions that are able to further shrink the graph and thus speed up the algorithm.

\subsubsection{Articulation Points}
\label{ss:ap}

Let $\phi \in V$ be an articulation point in $G$ whose removal disconnects the graph into multiple connected components. For any of these components that does not contain any terminals, we show that all vertices in the component can be contracted into $\phi$.

\begin{lemma} \label{lem:ap}
  For an articulation point $\phi$ whose removal disconnects the graph $G$ into multiple connected components $(G_1, \dots, G_p)$ and a component $G_i$ with $i \in \{1,\dots,p\}$ that does not contain any terminals, no edge in $G_i$ or connecting $G_i$ with $\phi$ can be part of $\cut(G)$. \end{lemma}
\begin{proof}
  Let $e$ be an edge that connects two vertices in $\{V_i \cup \phi\}$. Assume $e \in \cut(G)$, \ie $e$ is part of the minimum multiterminal cut of $G$. This means that vertices in $\{V_i \cup \phi\}$ are not all in the same block. By changing the block affiliation of all vertices in $\{V_i \cup \phi\}$ to $\vopt(\phi)$ we can remove all edges connecting vertices in $\{V_i \cup \phi\}$ from the multiterminal cut, thus decrease the weight of the multiterminal cut by at least $w(e)$. As $\phi$ is an articulation point, $G_i$ is only connected to the rest of $G$ through $\phi$ and thus no new edges are introduced to the multiterminal cut. This is a contradiction to the minimality of $\cut(G)$, thus no edge $e$ that connects two vertices in $\{V_i \cup \phi\}$ is in the minimum multiterinal cut $\cut(G)$.
\end{proof}

Using Lemmas~\ref{lem:cont}~and~\ref{lem:ap} we can contract all components that contain no terminals into the articulation point $\phi$. All articulation points of a graph can be found in linear time using an algorithm by Tarjan and Vishkin~\cite{tarjan1985efficient} based on depth-first search. The algorithm performs a depth-first search and checks in the backtracking step whether for a vertex $v$ there exists an alternative path from the parent of $v$ to every of descendant of $v$. If there is no alternative path, $v$ is an articulation point in $G$.

\subsubsection{Equal Neighborhoods}

In many cases, the resulting graph of the reductions contains groups of vertices that are connected to the same neighbors. If the neighborhood and respective edge weights of two vertices are equal, we can use Lemmas~\ref{lem:cont}~and~\ref{lem:nbrhd} to contract them into a single vertex.

\begin{lemma} \label{lem:nbrhd}
  For two vertices $v_1$ and $v_2$ with $\{N(v_1) \backslash v_2\} = \{N(v_2) \backslash v_1\}$ where for all $v \in N(v_1) \backslash v_2$, $w(v_1, v) = w(v_2, v)$, there is at least one minimum multiterminal cut where $\vopt(v_1) = \vopt(v_2)$. 
\end{lemma}
\begin{proof}
  Let $C$ be a partitioning of the vertices in $G$ with $C(v_1) \not = C(v_2)$, let $\zeta$ be the corresponding cut, where $e=(u,v) \in \zeta$, if $C(u) \neq C(v)$ and let $cw(v)$ be the total weight of edges in $\zeta$ incident to a vertex $v \in V$. W.l.o.g. let $v_2$ be the vertex with $cw(v_2) \geq cw(v_1)$. We analyze this in two steps: We assume that when moving $v_2$ to $C(v_1)$ that all edges incident to $v_2$ in its old location are removed from $\zeta$, which drops the weight of $\zeta$ by $cw(v_2)$ and then all edges incident to $v_2$ in its new location are added to $\zeta$, which is exactly $cw(v_1)$ by the conditions of the lemma. Thus the weight of $\zeta$ changes by $cw(v_1) - cw(v_2) \le 0$. If the edge $e_{12} = (v_1, v_2)$ exists, both $cw(v_1)$ and $cw(v_2)$ are furthermore decreased by $w(e_{12})$, as the edge connecting them is not a cut edge anymore. As we only moved the block affiliation of $v_2$, the only edges newly introduced to $\zeta$ are edges incident to $v_2$. Thus, the total weight of the multiterminal cut was not increased by moving $v_1$ and $v_2$ into the same block and we showed that for each cut $\zeta$, in which $C(v_1) \not = C(v_2)$ there exists a cut of equal or better value in which $v_1$ and $v_2$ are in the same block. Thus, there exists at least one multiterminal cut where $\vopt(v_1) = \vopt(v_2)$. 
\end{proof}

We detect equal neighborhoods for all vertices with neighborhood size smaller or equal to a constant $c_{N}$ using two linear time routines. To detect neighboring vertices $v_1$ and $v_2$ with equal neighborhood, we sort the neighborhood vertex IDs including edge weights by vertex IDs (excluding the respective other vertex) for both $v_1$ and $v_2$ and check for equality. To detect non-neighboring vertices $v_1$ and $v_2$ with equal neighborhood, we create a hash of the neighborhood sorted by vertex ID for each vertex with neighborhood size smaller or equal to $c_{N}$. If hashes are equal, we check whether the condition for contraction is actually fulfilled. As the neighborhoods to sort only have constant size, they can be sorted in constant time and thus the procedures can be performed in linear time. We perform both tests, as the neighborhoods of neighboring vertices contain each other and therefore do not result in the same hash value; and non-neighboring vertices are not in each others neighborhood and therefore finding them requires checking the neighborhood of every neighbor, which results in a large search space. We set $c_{N} = 5$, as most equal neighborhoods we encountered are in vertices with neighborhood size $\leq 5$.

\subsubsection{Maximum Flow from Non-terminal Vertices}

Let $v$ be an arbitrary vertex in $V \backslash T$, \ie a non-terminal vertex of $G$. Let $(V_v, V \backslash V_v)$ be the largest minimum isolating cut of $v$ and the set of terminal vertices $T$. Lemma~\ref{lem:flow} shows that there is at least one minimum multiterminal cut $\cut(G)$ so that $\forall x \in V_v: \vopt(x) = \vopt(v)$ and thus $V_v$ can be contracted into a single vertex.

\begin{lemma} \label{lem:flow}
  Let $v$ be a vertex in $V \backslash T$. Let $(V_v, V \backslash V_v)$ be the largest minimum isolating cut of $v$ and the set of terminal vertices $T$ and let $\lambda(G, v, T)$ be the weight of the minimum isolating cut $(V_v, V \backslash V_v)$. There exists at least one minimum multiterminal cut $\cut(G)$ in which $\forall x \in V_v: \vopt(x) = \vopt(v)$. 
\end{lemma}
\begin{proof}
  As $(V_v, V \backslash V_v)$ is a minimum isolating cut with the terminal set as sinks, we know that no terminal vertex is in $V_v$. Assume that $\cut(G)$ cuts $V_v$, \ie there is a non empty vertex set $V_C \in V_v$ so that $\forall x \in V_C: \vopt(x) \not\in \vopt(v)$. We will show that the existance of such a vertex set contradicts the minimality of $\cut(G)$. Figure~\ref{fig:flow} gives an illustration of the vertex sets defined here.

    \begin{figure}
      \centering
      \includegraphics[width=5cm]{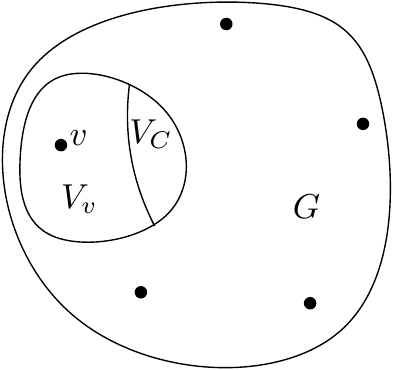}
      \caption{\label{fig:flow} Illustration of vertex sets in Lemma~\ref{lem:flow}}
    \end{figure}

    Due to the minimality of the minimum isolating cut, we know that $w(V_C, V_v \backslash V_C) \geq w(V_C, V \backslash V_v)$ (i.e. the connection of $V_C$ to the rest of $V_v$ is at least as strong as the connection of $V_C$ to $(V \backslash V_v)$), as otherwise we could remove $V_C$ from $V_v$ and find an isolating cut of smaller size.
  
  We now show that by changing the block affiliation of all vertices in $V_C$ to $\vopt(v)$, \ie removing all vertices from the set $V_C$, we can construct a multiterminal cut of equal or better cut value. By changing the block affiliation of all vertices in $V_C$ to $\vopt(v)$, we remove all edges connecting $V_C$ to $(V_v \backslash V_C)$ from $\cut(G)$ and potentially more, if there were edges in $\cut(G)$ that connect two vertices both in $V_C$. At most, the edges connecting $V_C$ and $(V \backslash V_v)$ are newly added to $\cut(G)$. As $w(V_C, V_v \backslash V_C) \geq w(V_C, V \backslash V_v)$, the cut value of $\cut(G)$ will be equal or better than previously. Thus, there is at least one multiterminal cut in which $V_C$ is empty and therefore $\forall x \in V_v: \vopt(x) = \vopt(v)$.
\end{proof}

We can therefore run a maximum $s$-$T$-flow from a non-terminal vertex to the set of all terminals $T$ and contract the source side of the largest minimum isolating cut into a single vertex. These flow problems can be solved embarassingly parallel, in which every processor solves an independent maximum $s$-$T$-flow problem for a different non-terminal vertex $v$.

While it is possible to run a flow problem from every vertex in $V$, this is obviously not feasible as it would entail excessive running time overheads. Promising vertices to use for maximum flow computations are either high degree vertices or vertices with a high distance from every terminal. High degree vertices are promising, as due to their high degree it is more likely that we can find a minimum isolating cut of size less than their degree. Vertices that have a high distance to all terminals are on 'the edge of the graph', potentially in a subgraph only weakly connected to the rest of the graph. Running a maximum flow then allows us to contract this subgraph. In every iteration, we run $5$ flow problems starting from high-distance vertices and $5$ flow problems starting from high-degree vertices. 

{}

\subsection{Vertex Branching}

When the \vcbase{}-algorithm is initialized, it only has a single problem containing the whole graph $G$. While independent minimum isolating cuts are computed in parallel, most of the shared-memory parallelism in \vcbase{} comes from the embarassingly parallel solving of different problems on separate threads. When branching, \vcbase{} selects the highest degree vertex that is adjacent to a terminal and branches on the heaviest edge connecting it to one of the terminals. The algorithm thus creates only up to two subproblems and is still not able to use the whole machine. 

We propose a new branching rule that overcomes these limitations by selecting the highest degree vertex incident to at least one terminal and use it to create multiple subproblems to allow for faster startup. Let $x$ be the vertex used for branching, $\{t_1,\dots,t_i\}$ for some $i \geq 1$ be the adjacent terminals of $x$ and $w_M$ be the weight of the heaviest edge connecting $x$ to a terminal. We now create up to $i + 1$ subproblems as follows: 

For each terminal $t_j$ with $j \in \{1,\dots,i\}$ with $w(x, t_j) + w(x, V \backslash T) > w_M$ create a new problem $P_j$ where edge $(x, t_j)$ is contracted and all other edges connecting $x$ to terminals are deleted. Thus in problem $P_j$, vertex $x$ belongs to block $\vopt(t_j)$. If $w(x, t_j) + w(x, V \backslash T) \leq w_M$, \ie the weight sum of the edges connecting $x$ with $t_j$ and all non-terminal vertices is not heavier than $w_M$, the assignment to block $\vopt(t_j)$ cannot be optimal and thus we do not need to create the problem $P_j$, also called \emph{pruning} of the problem:
\begin{lemma}
Let $G = (V,E)$ be a graph, $T \subseteq V$ be the set of terminal vertices in $G$, and $x \in V$ be a vertex that is adjacent to at least one terminal and for an $i \in \{1,\dots,|T|\}$ be the index of the terminal for which $e_i=(x,t_i)$ is the heaviest edge connecting $x$ with any terminal. Let $w_M$ be the weight of $e_i$. If there exists a terminal $t_j$ adjacent to $x$ with $j \in \{1,\dots,
|T|\}$ with $w(x, t_j) + w(x, V \backslash T) \geq w_M$, there is at least one minimum multiterminal cut $\cut(G)$ so that $\vopt(x) \not = j$, \ie $x$ is not in block $j$. 
\end{lemma}

\begin{proof}
  If $\vopt(x) = i$, \ie $x$ is in the block of the terminal it has the heaviest edge to, the sum of cut edge weights incident to $x$ is $\leq E(x) - w_M$, as edge $e_i$ of weight $w_M$ is not a cut edge in that case. If $\vopt(x) = j$, \ie $x$ is in the block of terminal $j$, the sum of cut edge weights incident to $x$ is $\geq E(x) - (w(x, V \backslash T) + w(x, t_j))$, as all edges connecting $x$ with other terminals than $t_j$ are guaranteed to be cut edges. As $w(x, t_j) + w(x, V \backslash T) \geq w_M$, even if all non-terminal neighbors of $x$ are in block $j$, the weight sum of incident cut edges is not lower than when $x$ is placed in block $i$. As the block affiliation of $x$ can only affect its incident edges, the cut value of every solution that sets $\vopt(x) = j$ would be improved or remain the same by setting $\vopt(x) = i$.
\end{proof}

If $w(x, V \backslash T) > w_M$ and $i < |T|$, we also create problem $P_{i+1}$, in which all edges connecting $x$ to a terminal are deleted. This problem represents the assignment of $x$ to a terminal that is not adjacent to it. We add each subproblem whose lower bound is lower than the currently best found solution $\bestwgt$ to the problem queue $\queue$. As we create up to $|T|$ subproblems, this allows for significantly faster startup of the algorithm and allows us to use the whole parallel machine after less time than before.

\subsection{Integer Linear Programming}

Integer Linear Programming can be used as an alternative to branch-and-reduce~\cite{henzinger2020shared} and for some problems this is faster than branch-and-reduce. We integrate the ILP formulation from the work of~\cite{henzinger2020shared} and include it directly into \vcbase{} as an alternative to branching. We give the ILP solver a time limit and if it is unable to find an optimal solution within the time limit, we instead perform a branch operation. In Section~\ref{ss:exp_ilp} we study which subproblems to solve with an ILP first.

\subsection{Improving Bounds with Greedy Optimization}
\label{ss:greedy}
The \vcbase{} algorithm prunes problems which can't result in a solution which is better than the best solution found so far. Therefore, even though it is a deterministic algorithm that will output the optimal result when it terminates, performing greedy optimization on intermediate solutions allows for more aggressive pruning of problems that cannot be optimal. Additionally, \vcbase{} has reductions that depend on the value of $\bestwgt(G)$ and can thus contract more vertices if the cut value $\bestwgt(G)$ is lower.

For a subproblem $H=(V_H, E_H)$ with solution $\rho$, the original graph $G = (V_G, E_G)$ and a mapping $\pi: V_G \rightarrow V_H$ that maps each vertex in $V_G$ to the vertex in $V_H$ that encompasses it, we can transfer the solution $\rho$ to a solution $\gamma$ of $G$ by setting the block affiliaton of every vertex $v \in V_G$ to $\gamma(v) := \pi(\rho(v))$. The cut value of the solution $w(\gamma)$ is defined as the sum of weights of the edges crossing block boundaries, \ie the sum of edge weights where the incident vertices are in differnet blocks. Let $\xi_i(V_G)$ be the set of all vertices $v \in V_G$ where $\gamma(v) = i$.

We introduce the following greedy optimization operators that can transform $\gamma$ into a better multiterminal cut solution $\gamma_{\text{IMP}}$ with $w(\gamma_{\text{IMP}}) < w(\gamma)$.

\subsubsection{Kernighan-Lin Local Search}

Kernighan and Lin~\cite{lin1973effective} give a heuristic for the traveling-salesman problem that has been adapted to many hard optimization problems~\cite{sandersschulz2013,traff2006direct,xu2005survey,dorigo2006ant}, where each vertex $v \in V_G$ is assigned a gain $g(v) = \max_{i \in \{i,\dots,|T|\}, i \not = \gamma(v)} \sum w(v, \xi_i(V_G)) - w(v, \xi_{\gamma(v)}(V_G))$, \ie the improvement in cut value to be gained by moving $v$ to another block, the best connected other block. We perform runs where we compute the gain of every vertex that has at least another neighbor in a different block and move all vertices with non-negative gain. Additionally, if a vertex $v$ has a negative gain, we store its gain and associated best connected other block. For any neighbor $u$ of $v$ that also has the same best connected other block, we check whether $g(w) + g(v) + 2 \cdot w(v, u) > 0$, \ie moving both $u$ and $v$ at the same time is a positive gain move. If it is, we perform the move.

\subsubsection{Pairwise Maximum Flow}

For any pair of blocks $1 \leq i < j \leq |T|$ where $w(\xi_i(V_G), \xi_j(V_G)) > 0$, \ie there is at least one edge from block $i$ to block $j$, we can create a maximum $s$-$t$ flow problem between them: we create a graph $F_{ij}$ that contains all vertices in $\xi_i(V_G)$ and $\xi_j(V_G)$ and all edges that connect these vertices. 

Let $H$ be the current problem graph created by performing reductions and branching on the original graph $G$. All vertices that are encompassed in the same vertex in problem graph $H$ as the terminals $i$ and $j$ are hereby contracted into the corresponding terminal vertex. We perform a maximum $s$-$t$-flow between the two terminal vertices and re-assign vertex assignments in $\gamma$ according to the minimum $s$-$t$-cut between them. As we only model blocks $\xi_i(V_G)$ and $\xi_j(V_G)$, this does not affect other blocks in $\gamma$. In the first run we perform a pairwise maximum flow between every pair of blocks $i$ and $j$ where $w(\xi_i(V_G), \xi_j(V_G)) > 0$ in random order. We continue on all pairs of blocks where $w(\xi_i(V_G), \xi_j(V_G))$ was changed since the end of the previous maximum flow iteration between them.

We first perform Kernigham-Lin local search until there is no more improvement, then pairwise maximum flow until there is no more improvement, followed by another run of Kernigham-Lin local search. As pairwise maximum flow has significantly higher running time, we spawn a new thread to perform the optimization if there is a CPU core that is not currently utilized.

\section{Fast Inexact Solving}
\label{s:inexact}

\vcbase{} in an exact algorithm, \ie when it terminates the output is guaranteed to be optimal. As the multiterminal cut problem is NP-complete~\cite{dahlhaus1994complexity}, it is not feasible to expect termination in difficult instances of the problem. Henzinger~\etal\cite{henzinger2020shared} report that their algorithm often does not terminate with an optimal result but runs out of time or memory and returns the best result found up to that point. Thus, it makes sense to relax the optimality constraint and aim to find a high-quality (but not guaranteed to be optimal) solution faster.

\begin{figure} 
  \centering
  \includegraphics[width=.35\textwidth]{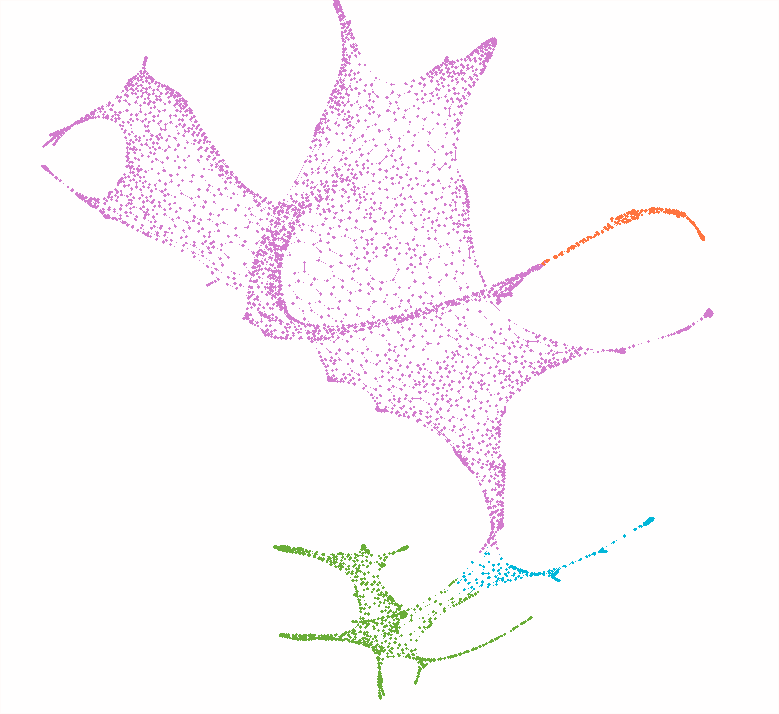}
  \includegraphics[width=.35\textwidth]{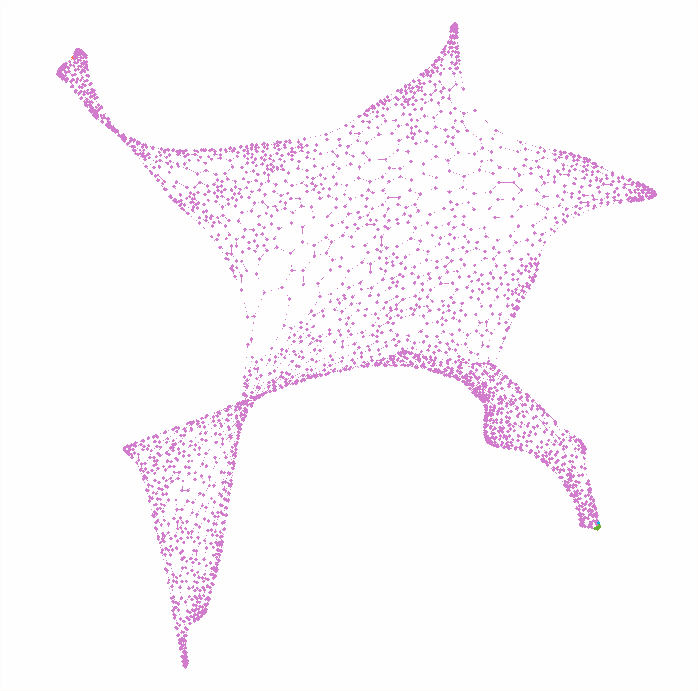}
  \caption{\label{fig:uk} Minimum multiterminal cut for graph \texttt{uk}~\cite{soper2004combined} and four terminals - on complete graph (left) and remaining graph at time of first branch operation (right), visualized using Gephi-0.9.2~\cite{bastian2009gephi}}
\end{figure}

A key observation herefor is that in many problems, most, if not all vertices that are not already contracted into a terminal at the time of the first branch, will be assigned to a few terminals whose weighted degree at that point is highest. See Figure~\ref{fig:uk} for an example with $4$ terminals (selected with high distance to each other) on graph \texttt{uk} from the Walshaw Graph Partitioning Archive~\cite{soper2004combined}. As we can see, at the time of the first branch (right figure), most vertices that are not assigned to the pink terminal in the optimal solution are already contracted into their respective terminals. The remainder is mostly assigned to a single terminal. As we can observe similar behavior in many problems, we propose the following heuristic speedup operations:

Let $\delta \in (0, 1)$ be a contraction factor and $T_H$ be the set of all terminals that are not yet isolated in graph $H$. In each branching operation on an intermediate graph $H$, we delete all edges around the $\lceil \delta \cdot |T_H| \rceil$ terminals with lowest degree. Additionally, we contract all vertices adjacent to the highest degree terminal that are not adjacent to any other terminal into the highest degree terminal. This still allows us to find all solutions in which no more vertices were added to the lowest degree terminals and the adjacent vertices are in the same block as the highest degree terminals.

Additionally, in a branch operation on vertex $v$, we set a maximum branching factor $\beta$ and only create problems where $v$ is contracted into the $\beta$ adjacent terminals it has the heaviest edges to and one problem in which it is not contracted into either adjacent terminal. This is based on the fact that all other edges connecting $v$ to other terminals will be part of the multiterminal cut and the greedy assumption that it is likely that the optimal solution does not contain at least one of these heavy edges. By default, we set $\delta = 0.1$ and $\beta = 5$.

\section{Experiments and Results} \label{s:experiments}

We now perform an experimental evaluation of the proposed work.
This is done in the following order: first we analzye the impact of different reductions introduced in the work of Henzinger~\etal\cite{henzinger2020shared} and in this work. We then analyze which subproblems to solve using integer linear programming and then compare the results of \vcbase{} with our exact and inexact algorithms on a variety of graphs from different sources. Here, \vcbase{} denotes the algorithm of Henzinger~\etal\cite{henzinger2020shared}, \exact{} denotes the exact version of our algorithm and \inexact{} denotes the heuristic algorithm proposed in Section~\ref{s:inexact}

We implemented the algorithms using \CC-17 and compiled all code using g++ version 7.3.0 with full optimization (\texttt{-O3}). Our experiments are conducted on two machine types. Machine A is a machine with two Intel Xeon E5-2643v4 with $3.4$ GHz with $6$ CPU cores each and $1.5$ TB RAM in total. Machine B is a machine in \ifDoubleBlind{\emph{name of cluster removed, as this is double blind}}\else{the Vienna Scientific Cluster}\fi{} with two Intel Xeon E5-2650v2 with $2.6$GHz with $8$ CPU cores each and $64$ GB RAM in total. We limit the maximum amount of memory used for each problem to $32$ GB. ILP problems are solved using Gurobi 8.1.0. When we report a mean result we give the geometric mean as problems differ significantly in cut size and time.  Our code is freely available under the permissive MIT license\ifDoubleBlind{\footnote{Link removed, as this submission is double blind}}\else{\footnote{\url{https://github.com/alexnoe/VieCut}}\fi{}. 

To evaluate the performance of different multiterminal cut algorithms, we use a wide variety of graphs from different sources. We re-use a large subset of the map, social and web graphs graphs used by Henzinger~\etal\cite{henzinger2020shared}. Additionally, we add numerical graphs from the Walshaw Graph Partitioning Benchmark~\cite{soper2004combined} and a set of graphs from the $10^{th}$ DIMACS implementation challenge~\cite{bader2013graph} and the SuiteSparse Matrix Collection (formerly UFSparse Matrix Collection)~\cite{davis2011university}. Table~\ref{t:graphs} gives an overview over the graphs used in this work. A table with properties of the instances be found in Appendix~\ref{app:instances}.

As the instances generally do not have any terminals, we find random vertices that have a high distance from each other in the following way: we start with a random vertex $r$, run a breadth-first search starting at $r$ and select the vertex $v$ encountered last as first terminal. While the number of terminals is smaller than desired, we add another terminal by running a breadth-first search from all current terminals and adding the vertex encountered last to the list of terminals. We then run a bounded-size breadth-first search around each terminal to create instances where the minimum multiterminal cut does not have $k - 1$ blocks consisting of just a single vertex each. This results in problems in which well separated clusters of vertices are partitioned and the task consists of finding a partitioning of the remaining vertices in the boundary regions between already partitioned blocks. This relates to clustering tasks, in which well separated clusters are labelled and the task consists of labelling the remaining vertices inbetween. Additionally, we use a subset of the generated instances of Henzinger~\etal\cite{henzinger2020shared} to compare our work to \vcbase{}. 

In order to compare different algorithms, we use \emph{performance profiles}~\cite{dolan2002benchmarking}.
These plots relate the cut size of all algorithms to the corresponding cut size produced by each algorithm.
More precisely, the $y$-axis shows $\#\{\text{objective} \leq \tau * \text{best} \} / \# \text{instances}$, where objective corresponds to
the result of an algorithm on an instance and best refers to the best result of any algorithm shown within the plot.
The parameter $\tau\geq 1$ in this equation is plotted on the $x$-axis.
For each algorithm, this yields a non-decreasing, piecewise constant function.
Thus, if we are interested in the number of instances where an algorithm is the best, we only need to look at $\tau = 1$.

\subsection{Reductions}
\label{ss:exp_red}

\begin{figure}[t!]
  \centering
  \includegraphics[width=.95\textwidth]{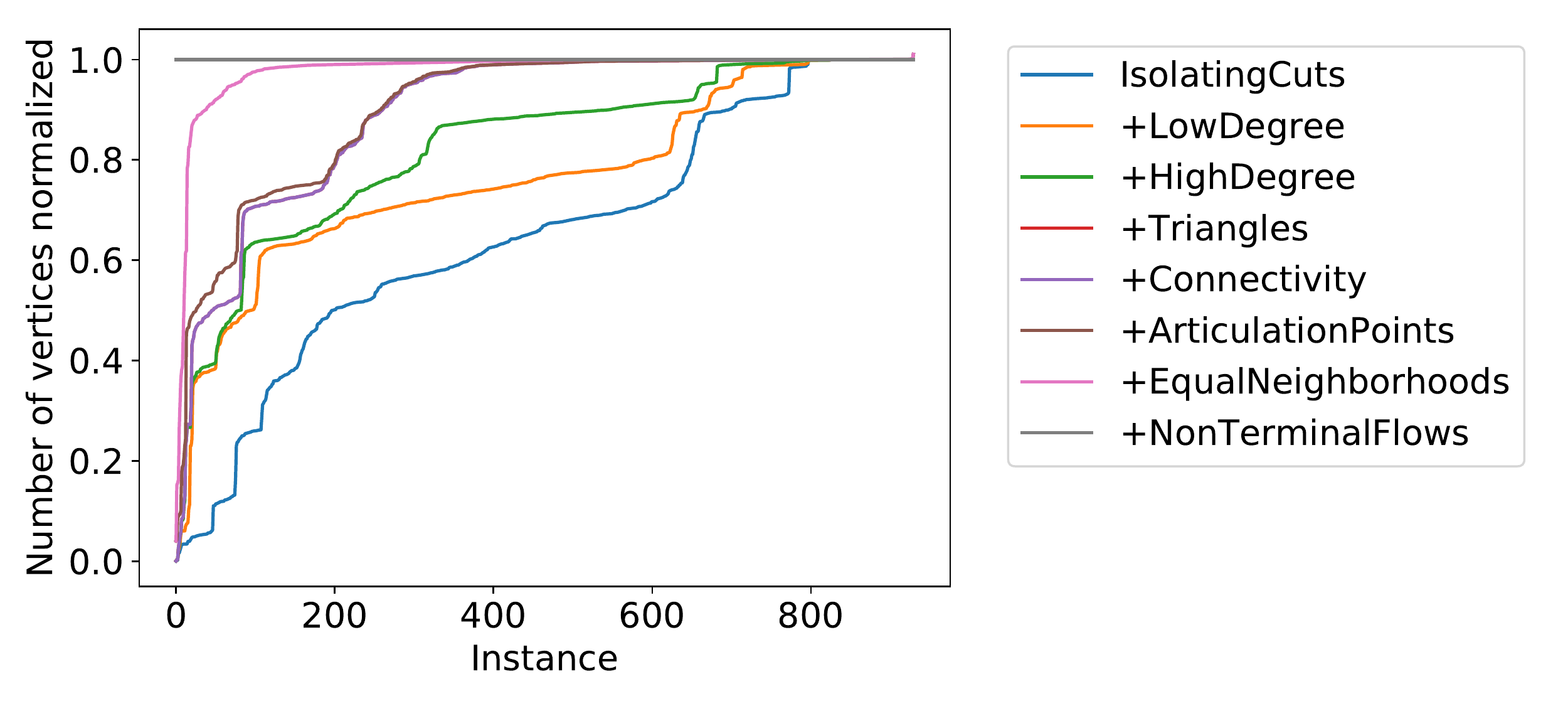}
  \caption{\label{fig:kernel} Number of vertices in graph after reductions are finished, normalization by (\# vertices remaining with all reductions / \# vertices remaining in variant) and sorted by normalized value.}
\end{figure}

We first analyze the impact of the different reductions on the size of the graph at the time of first branch. For this, we run experiments on all graphs in Appendix~\ref{app:instances} with $k=\{4,8,10\}$ terminals and $10\%$ of all vertices added to the terminals on machine B. On these instances, we run subsets of all contractions exhaustively and check which factor of vertices remain in the graph. A value of $1$ thus indicates that the reductions were unable to find any edges to contract, a value close to $0$ shows that almost no vertices remain and the resulting problem is far smaller than the original problem. Figure~\ref{fig:kernel} gives result with $8$ different variants, starting with a version that only runs isolating cuts and adding one reduction family per version. For this, we sorted the reductions by their impact on the total running time. In Figure~\ref{fig:kernel}, we can see that using all reductions allows us to reduce the number of vertices by more than half in about half of all instances and can find a sizable number of reductions on almost any instance. 

We can see that running the local reductions in \vcbase{} are very effective on almost all instances. In average, \texttt{IsolatingCuts} reduce the number of vertices by $33\%$, \texttt{LowDegree} reduces the number of vertices in the remaining graph by $17\%$, \texttt{HighDegree} by $7\%$ and \texttt{Triangles} by $8\%$. In contrast, \texttt{Connectivity} only has a negligible effect, which can be explained by the fact that it contracts edges whose connectivity is larger than a value related to the difference of upper bound to total weight of deleted edges. As there are almost no deleted edges in the beginning, this value is very high and almost no edge has high enough connectivity.

Out of the new reductions that are not part of \vcbase, all find a significant amount of contractible edges on the graphs already contracted by the reductions included therein. In average, \texttt{ArticulationPoints} reduces the number of vertices on the already contracted graphs by $1.9\%$, \texttt{EqualNeighborhoods} reduces the number of vertices by $7.8\%$ and \texttt{NonTerminalFlows} reduces the number of vertices by $2.0\%$. However, there are some instances in which the newly introduced reductions reduce the number of vertices remaining by more than $99\%$.

\subsection{Integer Linear Programming}
\label{ss:exp_ilp}

In order to get all a wide variety of ILP problems, we run the \inexact{} algorithm on all instances in Appendix~\ref{app:instances} with $k=10$ terminals and $10\%$ of vertices added to the terminals on machine B. As \inexact{} removes low-degree terminals and contracts edges, we have subproblems with very different sizes and numbers of terminals. In this experiment, whenever \inexact{} chooses between branching and ILP on graph $G$, we select a random integer $r \in (1, $\numprint{200000}$)$. We use this random integer, as we want to have problems of all different sizes and using a hard limit would result in many instances just barely below that size limit. We select $200000$ edges as the maximum, as we did not encounter any larger instances in which the ILP was solved to optimality in the allotted time. If $|E| < r$, the problem is solved with ILP, otherwise the algorithm branches on a vertex incident to a terminal. The timeout is set to $60$ seconds.

\begin{figure}[t!]
  \centering
  \includegraphics[width=.95\textwidth]{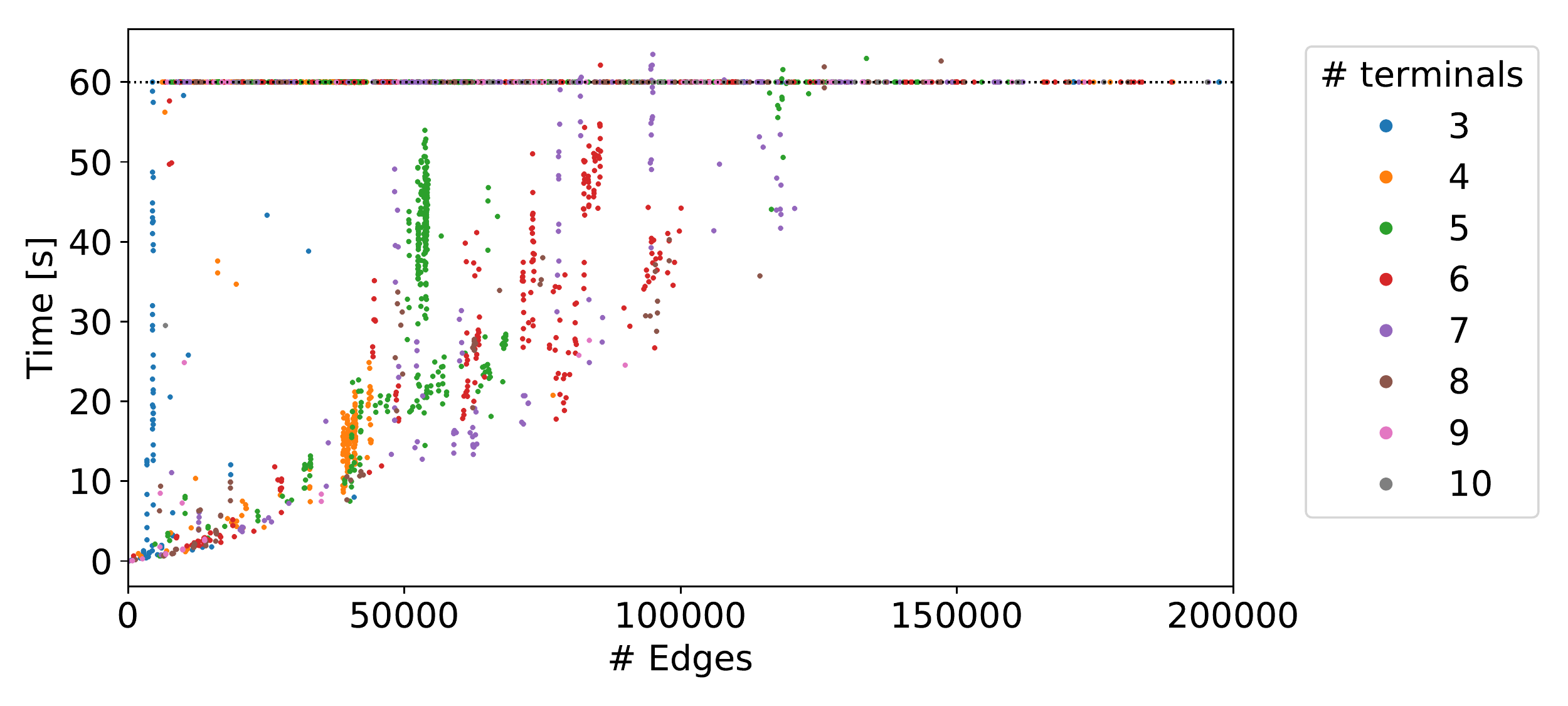}
  \caption{\label{fig:ilp} Running time of ILP subproblems in relation to $|E|$.}  
\end{figure}

Figure~\ref{fig:ilp} shows the time needed to solve the ILP problems in relation to the number of edges in the graph. We can see that there is a strong correlation between problem size and total running time, but there are still a large number of outliers that cannot be solved in the allotted time even though the instances are rather small. In the following, we set the limit to \numprint{50000} edges and solve all instances with fewer than \numprint{50000} edges with an integer linear program. If the instance has at least \numprint{50000} edges, we branch on a vertex incident to a terminal and create more subproblems.

\subsection{Comparison to \vcbase}
\label{ss:cpo}

We use the experiment of Section~$8.7$ in the work of Henzinger~\etal\cite{henzinger2020shared} to compare \exact{} to \vcbase{} on the instances used in their work. The experiment uses a set of large social and web graphs with pre-defined clusters and $k=\{3,4,5,8\}$ terminals, where $10-25\%$ of vertices are marked as terminal vertices initially, a total of $160$ instances. We run the experiments on machine A using all $12$ cores and set the time limit to $600$ seconds. 

Out of $160$ instances, \vcbase{} terminates with an optimal result in $32$ instances, while \exact{} terminates with an optimal result in $46$ instances. Out of the $115$ instances that were not solved to optimality by both algorithms, \exact{} gives a better result on $75$ instances and the same result on the other $38$ instances. The geometric mean of results given by \exact{} and \inexact{} are both about $1.5\%$ lower than \vcbase{}. Note that in the experiments performed in~\cite{henzinger2020shared}, which uses a larger machine ($32$ cores) and has a timeout of $3600$ seconds, \vcbase{} has a geometric mean of about $0.1\%$ better than \vcbase{} in this work. The largest part of the improvement of \exact{} and \inexact{} over \vcbase{} is gained by the greedy optimization detailed in Section~\ref{ss:greedy}.
\begin{table}[b!] \centering
  \vspace*{.5cm}
  \caption{\label{t:probs} Result overview for large multiterminal cut problems on graphs from Appendix~\ref{app:instances}.}
  \begin{tabular}{llrrr}
    \toprule
    \# Terminals & & \vcbase & \exact & \inexact \\
    \midrule
    4 & Best Solution & \numprint{109} & \textbf{\numprint{183}} & \numprint{175} \\
    & Mean Solution & \numprint{161799} & \textbf{\numprint{159402}} & \numprint{159499} \\
    & Better Exact & \numprint{6} & \textbf{\numprint{94}} & --- \\
    \midrule
    5 & Best Solution & \numprint{81} & \textbf{\numprint{173}} & \numprint{158} \\
    & Mean Solution & \numprint{216191} & \textbf{\numprint{210928}} & \numprint{211090} \\
    & Better Exact & \numprint{6} & \textbf{\numprint{121}} & --- \\
    \midrule
    8 & Best Solution & \numprint{42} & \numprint{139} & \textbf{\numprint{175}} \\
     & Mean Solution & \numprint{346509} & \numprint{331112} & \textbf{\numprint{330856}} \\
     & Better Exact & \numprint{2} & \textbf{\numprint{162}} & --- \\
    \midrule
    10 & Best Solution & \numprint{37} & \numprint{129} & \textbf{\numprint{173}} \\
     & Mean Solution & \numprint{412138} & \numprint{392561} & \textbf{\numprint{391822}} \\
     & Better Exact & \numprint{1} & \textbf{\numprint{165}} & --- \\
    \bottomrule
    
  \end{tabular}
    \vspace*{.5cm}
\end{table}

Figure~\ref{fig:pp53} shows the performance profile of this experiment. We can see that both \exact{} and \inexact{} are almost always optimal or very close to it. In contrast, \vcbase{} gives noticably worse results on about $20\%$ of instances and more than $5\%$ worse results on $10\%$ of all instances. 

\subsection{Large Multiterminal Cut Problems}
\label{ss:mtc}

\begin{figure}[t!]
  \centering
  \begin{subfigure}{.65\textwidth}
    \includegraphics[width=\linewidth]{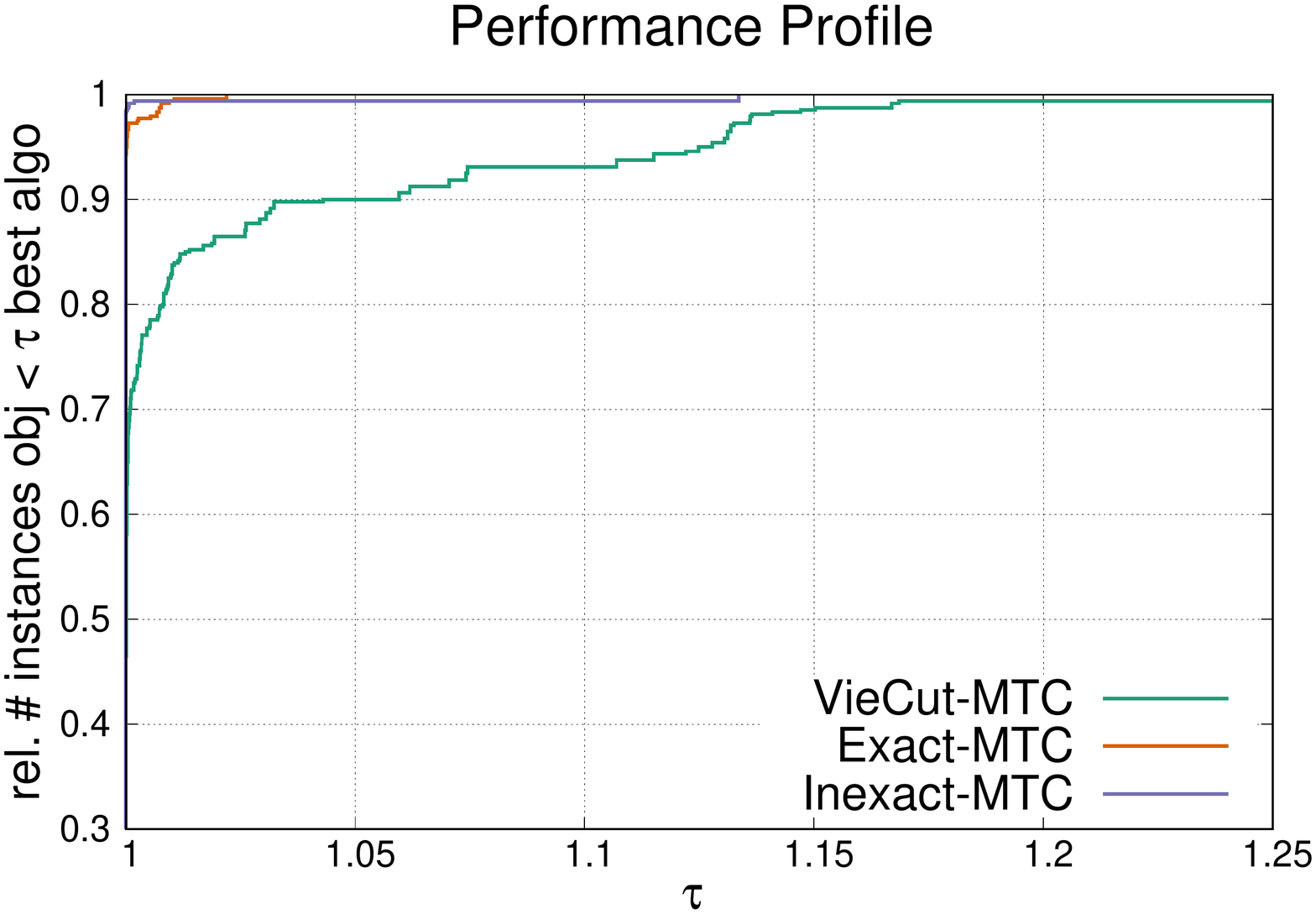}
    \subcaption{\label{fig:pp53} Instances of Section~\ref{ss:cpo}}
  \end{subfigure}\\%
  \begin{subfigure}{.65\textwidth}
    \includegraphics[width=\linewidth]{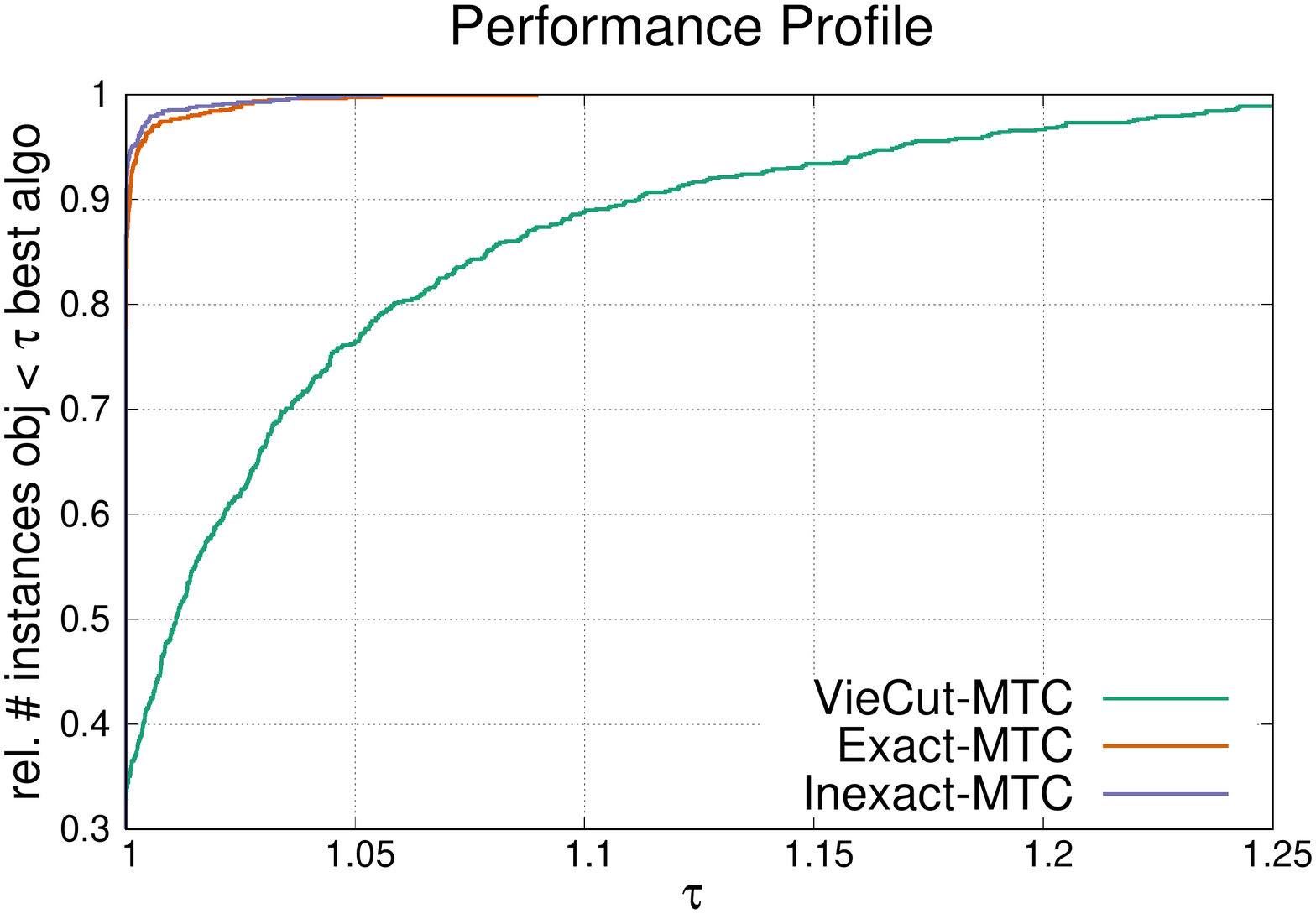}
    \subcaption{\label{fig:pp54} Instances of Section~\ref{ss:mtc}}
  \end{subfigure}
  \caption{Performance profiles for multiterminal cut algorithms}
  \vspace*{.5cm}
\end{figure}

We compare \vcbase{}, \exact{} and \inexact{} on all graphs with $k=\{4,5,8,10\}$ terminals and $10\%$ and $20\%$ of vertices added to the terminal. For each combination of graph, number of terminals and factor of vertices in terminal, we create three problems with random seeds $s=\{0,1,2\}$. Thus, we have a total of $816$ problems. We set the time limit per algorithm and problem to $600$ seconds. We run the experiment on machine A using all $12$ CPU cores. If the algorithm does not terminate in the allotted time or memory limit, we report the best intermediate result. Note that is a soft limit, in which the algorithm finishes the current operation and aims to terminate gracefully if the time or memory limit is reached. As many of these are very large instances, most instances in this section are not solved to optimality.

Table~\ref{t:probs} gives an overview of the results. For each algorithm, we give the number of times, where it gives the best (or shared best) solution over all algorithms; the geometric mean of the cut value; and for \vcbase{} and \exact{} the number of instances in which they have a better result than the respective other. In all instances, in which \vcbase{} and \exact{} terminate with the optimal result, \inexact{} also gives the optimal result. We can see that in the problems with $4$ and $5$ terminals, \exact{} slightly outperforms \inexact{} both in number of best results and mean solution value. In the problems with $8$ and $10$ terminals, \inexact{} has slightly better results in average. Thus, disregarding the optimality constraint can allow the algorithm to give better solutions faster especially in hard problems with a large amount of terminals. 

However, both new algorithms outperform \vcbase{} on almost all instances where not all algorithms give the same result. Here, \exact{} gives a better result than \vcbase{} in $66\%$ of all instances, while \vcbase{} gives the better result in only $2\%$ of all instances. As most problems do not terminate with an optimal result, we are unable to say how far the solutions are from the globally optimal solution. Note that \inexact{} gives an optimal result in all instances in which all algorithms terminate. Figure~\ref{fig:progress} shows the progress of the best solution for the algorithms in a set of problems. For both \exact{} and \inexact{} we can see large improvements to the cut value when the local search algorithm is finished on the first subproblem. In contrast, \vcbase{} has more small step-by-step improvements and generally gives worse results.  

Figure~\ref{fig:pp54} shows the performance profile for the instances in this section. Here we can see that \vcbase{} has significantly worse results on a large subset of the instances, with more than $10\%$ of instances where the result is worse by more than $10\%$. Also, on a few instances, the results given by \exact{} and \inexact{} differ significantly. In general, both of them outperform \vcbase on most instances that are not solved to optimality by every algorithm.

\section{Conclusion}

In this paper, we give a fast parallel solver that gives high-quality solutions for large multiterminal cut problems. We give a set of highly-effective reduction rules that transform an instance into a smaller equivalent one. Additionally, we directly integrate an ILP solver into the algorithm to solve subproblems well suited to be solved using an ILP; and develop a flow-based local search algorithm to improve a given optimal solution. These optimizations significantly increase the number of instances that can be solved to optimality and improve the cut value of multiterminal cuts in instances that can not be solved to optimality. Our algorithm gives better solutions in more than two thirds of these instances, often improving the result by more than $5\%$ on hard instances. Additionally, we give an inexact algorithm for the multiterminal cut problem that aggressively shrinks the graph instances and is able to even outperform the exact algorithm on many of the hardest instances that are too large to be solved to optimality while still giving the exact solution for most easier instances. Important future work consists of improving the scalability of the algorithm by giving a distributed memory version.

\bibliographystyle{plainurl}
\bibliography{paper,oldref}


\begin{appendix}
\newpage
\section{Instances}
\label{app:instances}

\begin{table}[h!] \centering
   \caption{\label{t:graphs}Large Real-world Benchmark Instances}
   \begin{tabular}{lrrr}
     \toprule
     Graph & Source & $n$& $m$\\
     \midrule
     \multicolumn{4}{c}{Map Graphs} \\
     \midrule
     ak2010 & \cite{bader2013graph} & \numprint{45292} & $109K$\\
     ca2010 & \cite{bader2013graph} & $710K$ & $1.74M$\\
     ct2010 & \cite{bader2013graph} & \numprint{67578} & $168K$\\
     de2010 & \cite{bader2013graph} & \numprint{24115} & \numprint{58028}\\
     hi2010 & \cite{bader2013graph} & \numprint{25016} & \numprint{62063}\\
     luxembourg.osm & \cite{davis2011university} & $115K$ & $120K$ \\
     me2010 & \cite{bader2013graph} & \numprint{69518} & $168K$\\
     netherlands.osm & \cite{davis2011university} & $2.22M$ & $2.44M$ \\
     nh2010 & \cite{bader2013graph} & \numprint{48837} & $117K$\\
     nv2010 & \cite{bader2013graph} & \numprint{84538} & $208K$\\
     ny2010 & \cite{bader2013graph} & $350K$ & $855K$ \\
     ri2010 & \cite{bader2013graph} & \numprint{25181} & \numprint{62875}\\
     sd2010 & \cite{bader2013graph} & \numprint{88360} & $205K$\\
     vt2010 & \cite{bader2013graph} & \numprint{32580} & \numprint{77799}\\

     \midrule
     \multicolumn{4}{c}{Social, Web and Numerical Graphs}\\
     \midrule
     598a & \cite{soper2004combined} & $111K$ & $742K$ \\
     astro-ph & \cite{davis2011university} & \numprint{16706} & $121K$\\
     bcsstk30 & \cite{soper2004combined} & \numprint{28924} & $1.01M$\\
     ca-CondMat & \cite{davis2011university} & \numprint{23133} & \numprint{93439}\\
     caidaRouterLevel & \cite{davis2011university} & $192K$ & $609K$\\
     citationCiteseer & \cite{davis2011university} & $268K$ & $1.16K$\\
     cit-HepPh & \cite{davis2011university} & \numprint{34546} & $422K$\\
     cnr-2000 & \cite{davis2011university} & $326K$ & $2.74M$\\
     coAuthorsCiteseer & \cite{davis2011university} & $227K$ & $814K$ \\
     cond-mat-2005 & \cite{davis2011university} & \numprint{40421} & $176K$\\
     coPapersCiteseer & \cite{davis2011university} & $434K$ & $16.0M$\\
     cs4 & \cite{soper2004combined} & \numprint{22499} & \numprint{43858} \\
     eu-2005 & \cite{BoVWFI} & $862K$ & $16.1M$\\
     fe\_body & \cite{soper2004combined} & \numprint{45087} & $164K$ \\
     higgs-twitter & \cite{davis2011university} & $457K$ & $14.9M$\\
     in-2004 & \cite{BoVWFI} & $1.38M$ & $13.6M$\\
     NACA0015 & \cite{davis2011university} & $1.04M$ & $3.11M$\\
     uk-2002 & \cite{BoVWFI} & $18.5M$ & $261M$ \\
     venturiLevel3 & \cite{davis2011university} & $4.03M$ & $8.05M$ \\
     vibrobox & \cite{soper2004combined} & \numprint{12328} & $165K$\\
     \bottomrule
  \end{tabular}
\end{table}
     
\vfill{}
\newpage

\section{Additional Figures}
\label{app:figs}

\begin{figure}[h!]
  \centering
  \begin{subfigure}{.49\textwidth}
    \includegraphics[width=\linewidth]{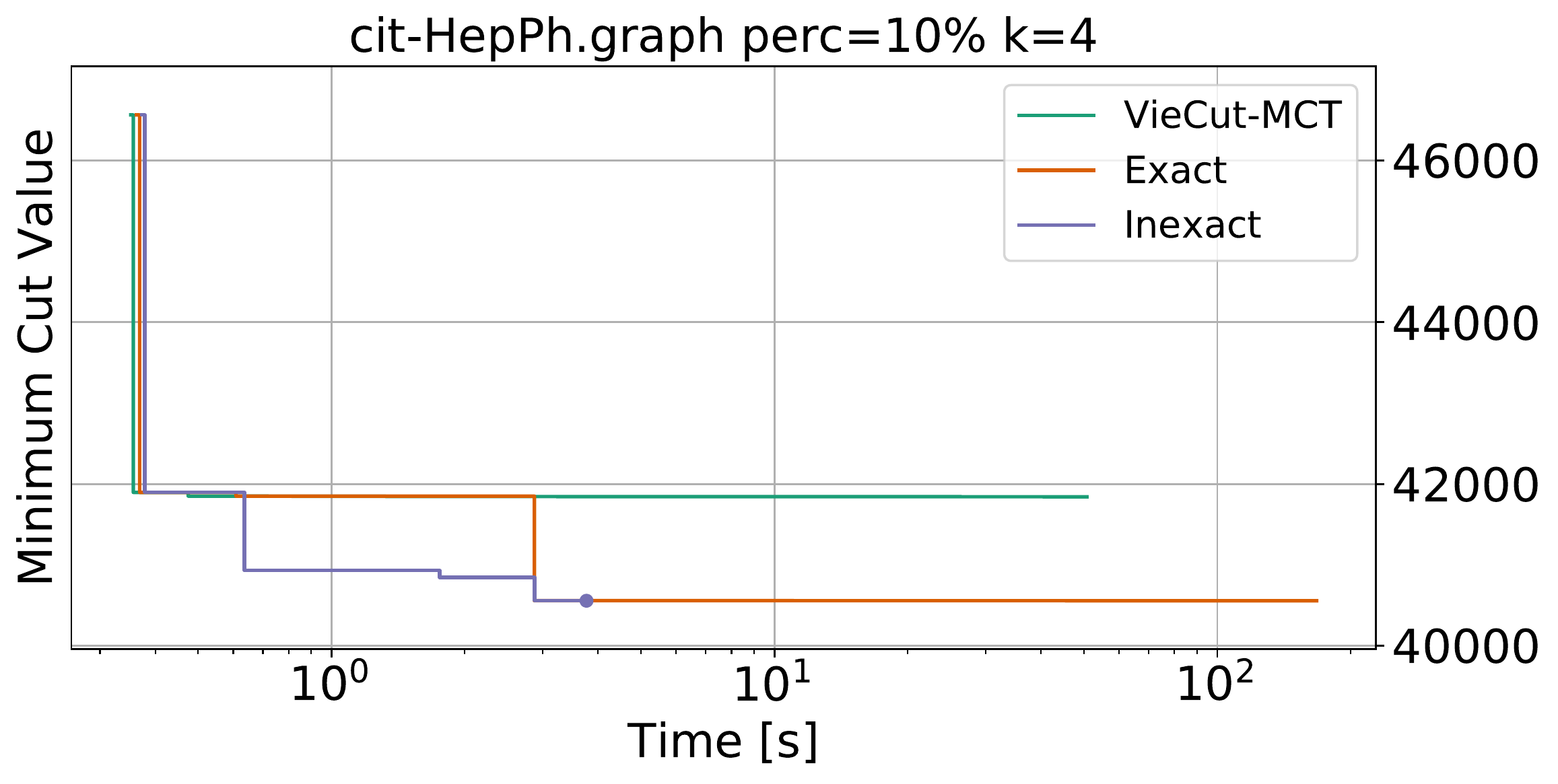}
  \end{subfigure}%
  \begin{subfigure}{.49\textwidth}
    \includegraphics[width=\linewidth]{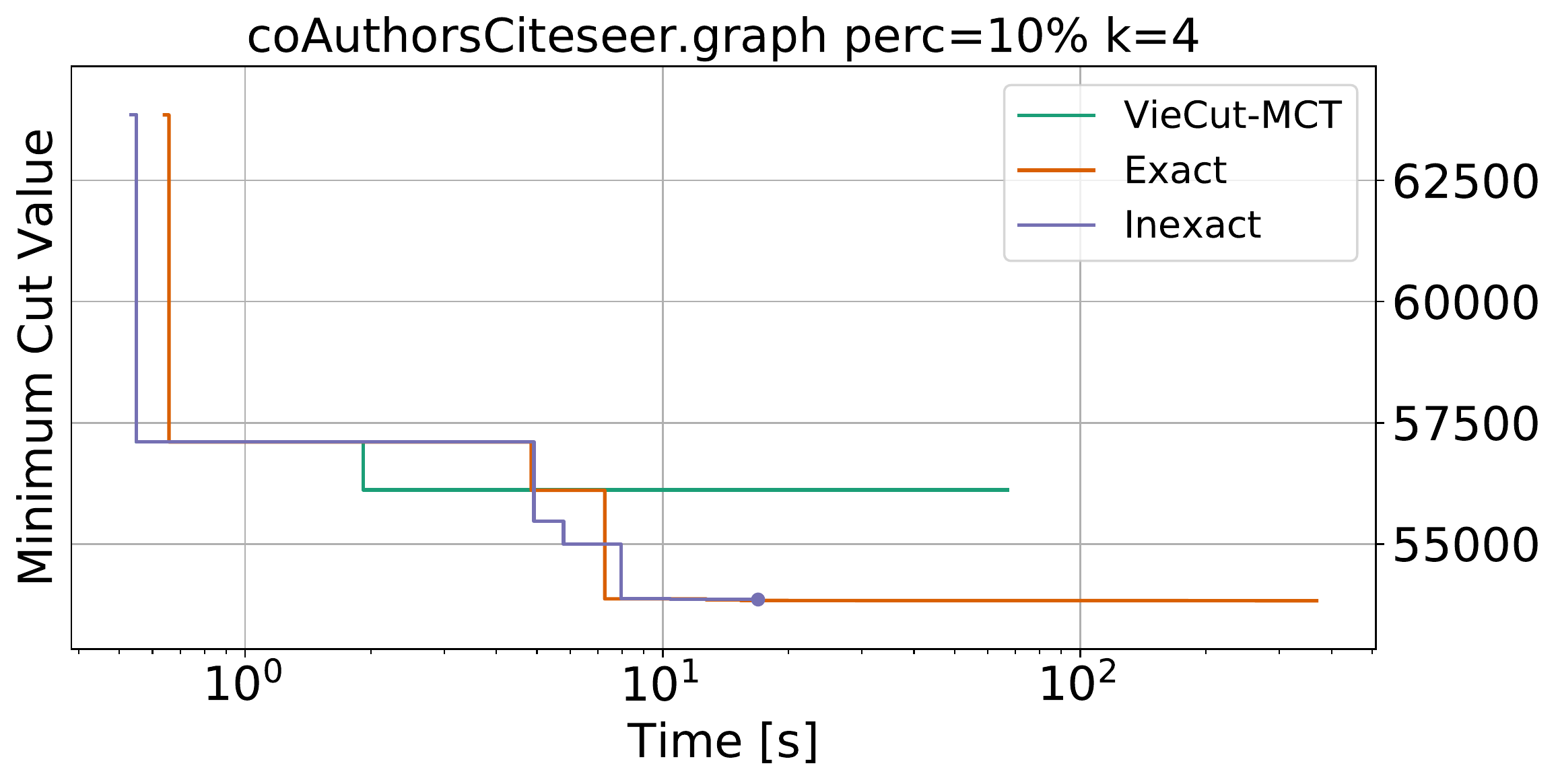}
  \end{subfigure}
  \begin{subfigure}{.49\textwidth}
    \includegraphics[width=\linewidth]{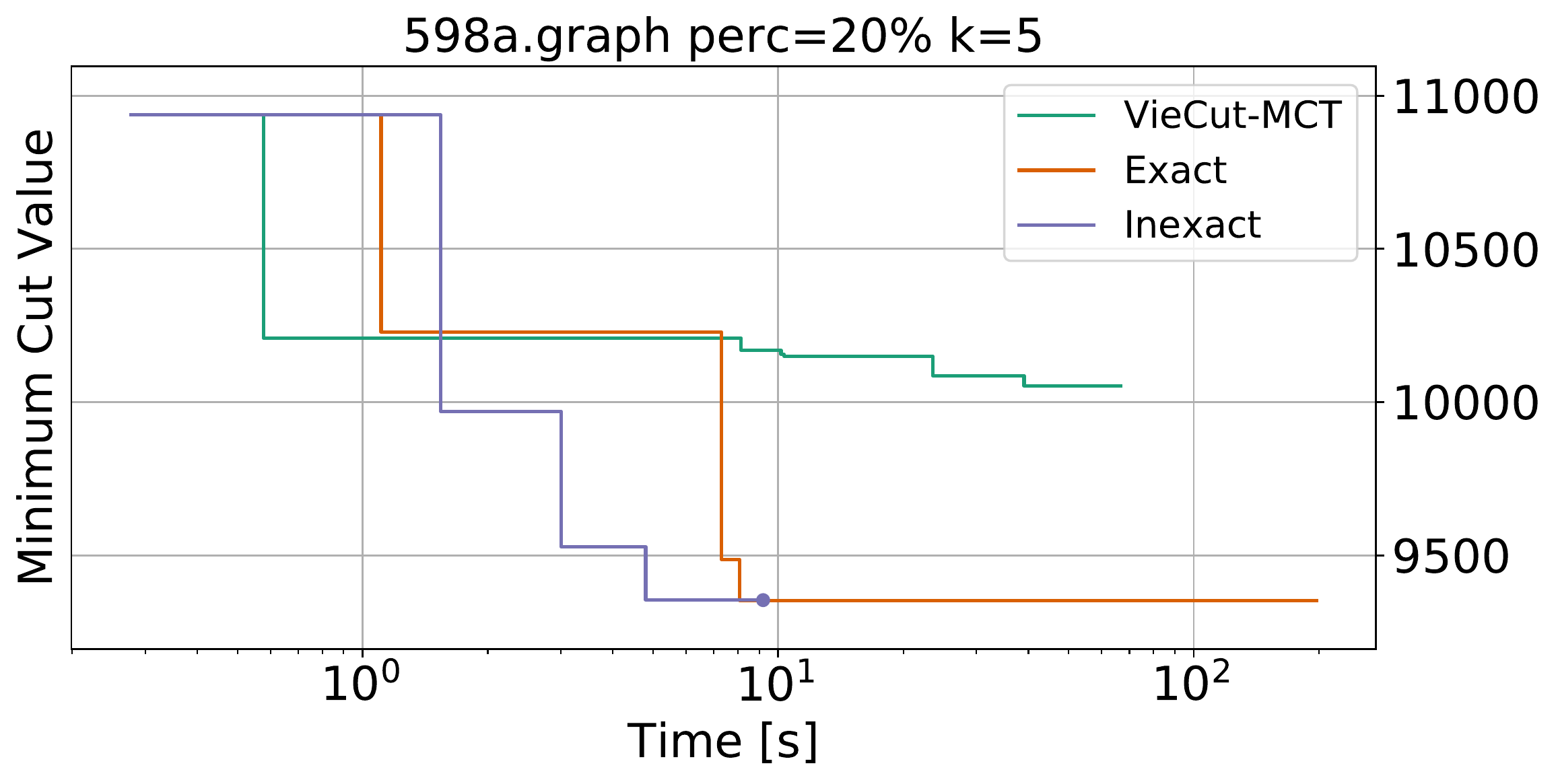}
  \end{subfigure}%
  \begin{subfigure}{.49\textwidth}
    \includegraphics[width=\linewidth]{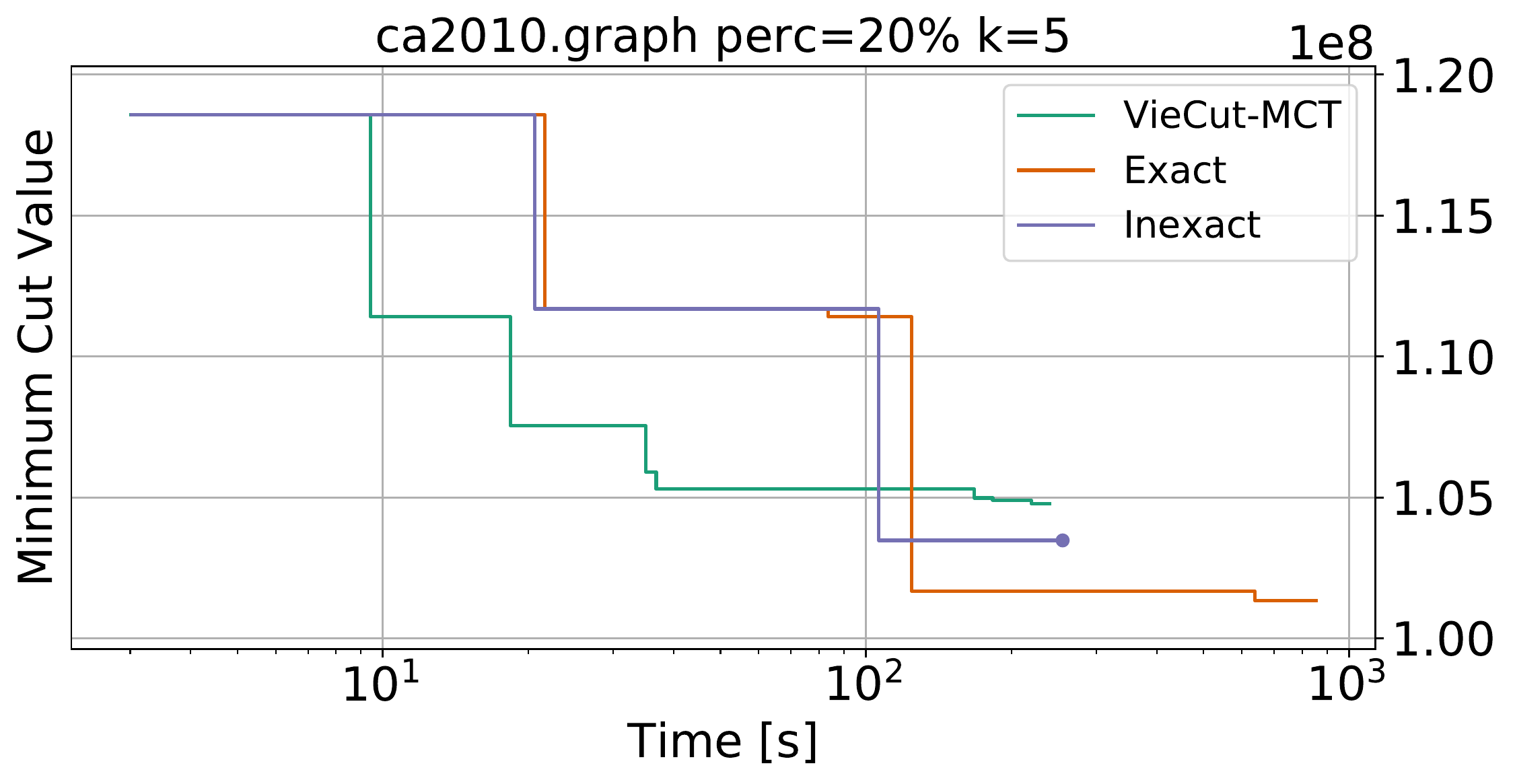}
  \end{subfigure}
  \begin{subfigure}{.49\textwidth}
    \includegraphics[width=\linewidth]{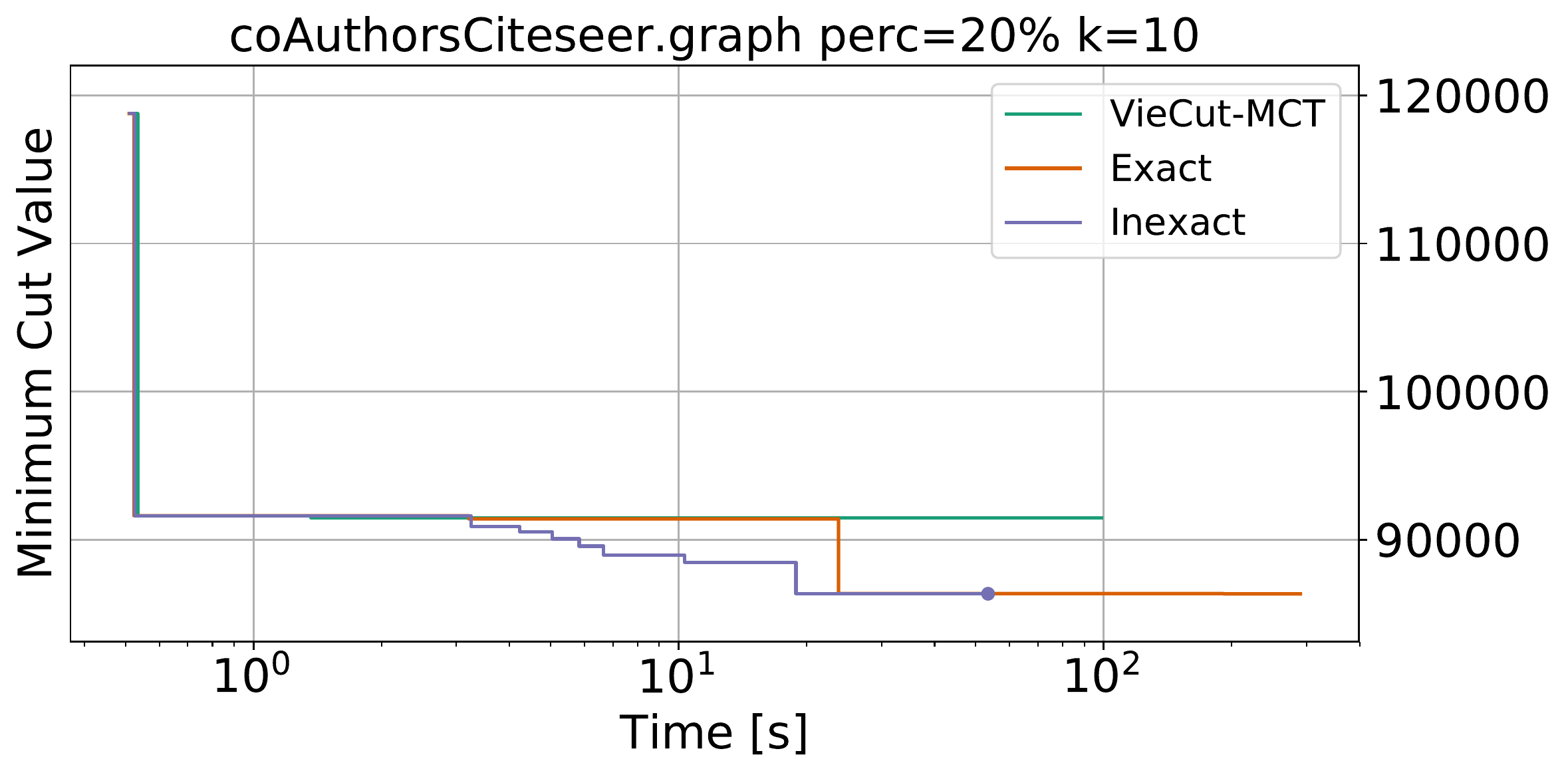}
  \end{subfigure}%
  \begin{subfigure}{.49\textwidth}
    \includegraphics[width=\linewidth]{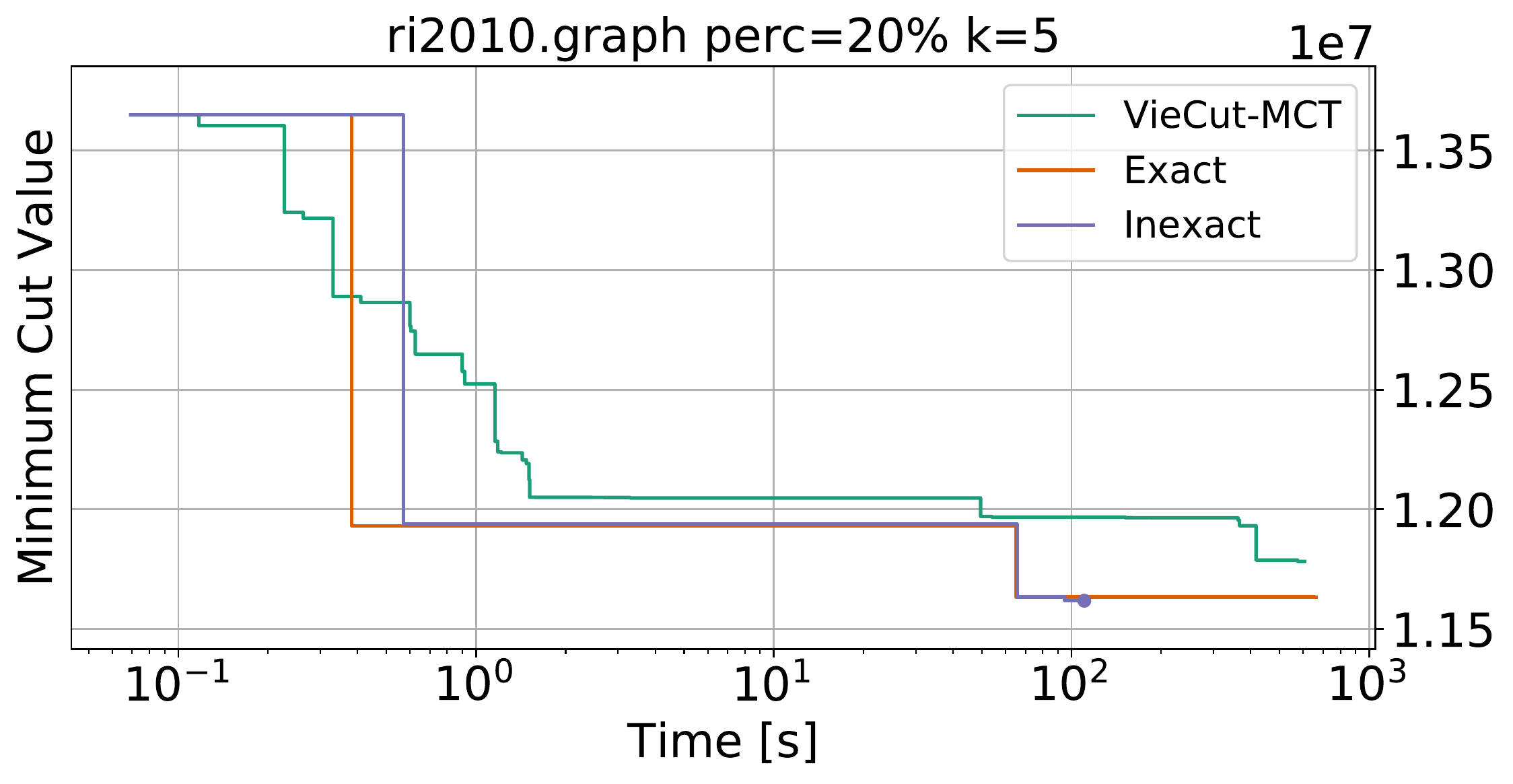}
  \end{subfigure}

  \caption{Progression of best result over time. Dot at end marks termination of algorithm.}
  \label{fig:progress}
\end{figure}

\end{appendix}

\end{document}